\documentclass[a4paper,11pt]{article}

\usepackage{breakcites}

\usepackage{fullpage}
\usepackage{times}
\usepackage{soul}
\usepackage{url}
\usepackage[utf8]{inputenc}
\usepackage[small]{caption}
\usepackage{graphicx}
\usepackage{amsmath}
\usepackage{booktabs}

\usepackage[ruled,vlined]{algorithm2e}

\usepackage{enumitem}

\usepackage{libertine}

\usepackage{amssymb,amsmath,amsfonts,amstext,amsthm}

\usepackage[breaklinks=true]{hyperref}
\usepackage[svgnames]{xcolor}
\usepackage[capitalise,nameinlink]{cleveref}
\hypersetup{colorlinks={true},linkcolor={DarkBlue},citecolor=[named]{DarkGreen}}

\usepackage{tikz}  
\usetikzlibrary{arrows}
\usetikzlibrary{patterns,snakes}
\usetikzlibrary{decorations.shapes}
\tikzstyle{overbrace text style}=[font=\tiny, above, pos=.5, yshift=5pt]
\tikzstyle{overbrace style}=[decorate,decoration={brace,raise=5pt,amplitude=3pt}]
\usetikzlibrary{shapes.geometric}
\usetikzlibrary{shapes,positioning}
\usetikzlibrary{calc,positioning}
\usetikzlibrary{shapes.multipart}

\definecolor{cadmiumgreen}{rgb}{0.0, 0.42, 0.24}

\usepackage{bbm}

\newtheorem{theorem}{Theorem}[section]
\newtheorem{corollary}[theorem]{Corollary}

\newtheorem{lemma}[theorem]{Lemma}

\theoremstyle{definition}

\usepackage{natbib}%

\usepackage{mathtools}
\usepackage{xcolor} 
\usepackage{tcolorbox}

\usepackage{authblk}

\usepackage{subcaption}

\newcommand{\dist}{\mathtt{dist}}

\newcommand{\MAX}{\text{\normalfont \texttt{MAX}}}
\newcommand{\AVG}{\text{\normalfont \texttt{AVG}}}

\begin{document}

\allowdisplaybreaks

\title{\bf The Distortion of Distributed Metric Social Choice}

\author[1]{Elliot Anshelevich}
\author[2]{Aris Filos-Ratsikas}
\author[3]{Alexandros A. Voudouris}

\affil[1]{Department of Computer Science, Rensselaer Polytechnic Institute, USA}
\affil[2]{Department of Computer Science, University of Liverpool, UK}
\affil[3]{School of Computer Science and Electronic Engineering, University of Essex, UK}

\renewcommand\Authands{ and }
\date{}

\maketitle

\begin{abstract}
We consider a social choice setting with agents that are partitioned into disjoint groups, and have metric preferences over a set of alternatives. Our goal is to choose a single alternative aiming to optimize various objectives that are functions of the distances between agents and alternatives in the metric space, under the constraint that this choice must be made in a distributed way: The preferences of the agents within each group are first aggregated into a representative alternative for the group, and then these group representatives are aggregated into the final winner. Deciding the winner in such a way naturally leads to loss of efficiency, even when complete information about the metric space is available. We provide a series of (mostly tight) bounds on the distortion of distributed mechanisms for variations of well-known objectives, such as the (average) total cost and the maximum cost, and also for new objectives that are particularly appropriate for this distributed setting and have not been studied before. 
\end{abstract}

\section{Introduction} \label{sec:intro}
The main goal of social choice theory~\citep{sen1986social} is to come up with outcomes that accurately reflect the collective opinions of individuals within a society. A prominent example is that of elections, where the preferences of voters over different candidates are aggregated into a single winner, or a set of winners in the case of committee elections. Besides elections, the abstract social choice theory setting, where a set of \emph{agents} express preferences over a set of possible \emph{alternatives} captures very broad decision-making application domains, such as choosing public policies, allocations of resources, or the most appropriate position to locate a public facility.

In the field of computational social choice, \citet{procaccia2006distortion} defined the notion of \emph{distortion} to measure the loss in an aggregate cardinal objective (typically the utilitarian social welfare), due to making decisions whilst having access to only limited ({\em ordinal}, in particular) information about the preferences of the agents, rather than their true cardinal values (or costs). Following their work, a lot of effort has been put forward to bound the distortion of social choice rules, with \citet{anshelevich2018approximating} and \citet{anshelevich2017randomized} being the first to considered settings with \emph{metric} preferences. In such settings, agents and alternatives are points in a metric space, and thus the distances between them (which define the costs of the agents) satisfy the triangle inequality. The metric space can be high-dimensional, and can be thought of as evaluating the proximity between agents and alternatives for different political issues or ideological axes (e.g., liberal to conservative, or libertarian to authoritarian). The distortion in metric social choice has received significant attention, with many variants of the main setting being considered over the recent years.  

In contrast to the {\em centralized} decision-making settings considered in the papers mentioned above, there are cases where it is logistically too difficult to aggregate the preferences of the agents directly, or different groups of agents play inherently different roles in the process. In such scenarios, the collective decisions have to be carried out in a \emph{distributed} manner, as follows. The agents are partitioned into groups (such as electoral districts, focus groups, or sub-committees), and the members of each group locally decide a single alternative that is representative of their preferences, without taking into account the agents of different groups. Then, the final outcome is decided based on properties of the group representatives, and not on the underlying agents within the groups; for example, the representatives act as agents themselves and choose an outcome according to their own preferences. However, since the representatives cannot perfectly capture all the information about the preferences of the agents (even when it is available in the group level), it is not surprising that choosing the final outcome this way may lead to loss of efficiency.

Motivated by this, \citet{FMV2020distributed} initiated the study of the deterioration of the social welfare in general normalized distributed social choice settings. They extended the notion of distortion to account for the information about the agents' preferences that is lost after the local decision step, and showed bounds on the distortion of max-weight mechanisms when the number of groups is given. Very recently, \citet{FV2020facility} considered the distributed distortion problem with metric preferences, and showed tight bounds on the distortion of (cardinal and ordinal) mechanisms under several restrictions: Their bounds (a) apply only to the real line metric, (b) concern the (average) social cost objective (the total distance between agents and the chosen alternative), and (c) are mainly limited to groups of agents of the same size. 

In this paper, we extend the results of \citet{FV2020facility} in all three axes: We provide bounds for (a) general metrics (including refined bounds for the line metric), (b) four different objectives (including the average total cost, the maximum cost, as well as two new objectives that are tailor-made and clearly motivated by the distributed nature of the setting), and (c) groups of agents that could vary in size. We paint an almost complete picture of the distortion landscape of distributed mechanisms when the agents have metric preferences.

\subsection{Our Contributions}
We consider a distributed, metric social choice setting with a set of agents and a set of alternatives, all of whom are located in a metric space. The preferences of the agents for the alternatives are given by their distances in the metric space, and as such they satisfy the triangle inequality. Furthermore, the agents are partitioned into a given set of {\em districts} of possibly different sizes. A distributed mechanism selects an alternative based on the preferences of the agents in two steps: first, each district selects a \emph{representative} alternative using some local aggregation rule, and then the mechanism uses \emph{only} information about the representatives to select the final winning alternative. 

The goal is to choose the alternative that optimizes some aggregate objective that is a function of the distances between agents and alternatives. In the main part of the paper we consider the following four cost minimization objectives, which can be defined as compositions of objectives applied over and within the districts: 
\begin{itemize}
\item The {\em average of the average agent distance in each district}\footnote{Note that this objective is not exactly equivalent to the well-known (average) {\em social cost} objective, defined as the (average) total agent distance over all districts. All of our results extend for this objective as well, by adapting our mechanisms to weigh the representatives proportionally to the district sizes. When all districts have the same size, $\AVG \circ \AVG$ coincides with the average social cost.} (denoted by $\AVG \circ \AVG$); 
\item The {\em average of the maximum agent distance in each district} (denoted by $\AVG \circ \MAX$);
\item The {\em maximum agent distance in any district} (denoted by $\MAX \circ \MAX$); 
\item The {\em maximum of the average agent distance in each district} (denote by $\MAX \circ \AVG$).
\end{itemize}
While $\AVG \circ \AVG$ and $\MAX \circ \MAX$ are adaptations of objectives that have been considered in the centralized setting, $\AVG \circ \MAX$ and $\MAX \circ \AVG$ are only meaningful in the context of distributed social choice. In particular, $\MAX \circ \AVG$  can be thought of as a fairness-inspired objective guaranteeing that no district has a very large cost, where the cost of a district is the average cost of its members. Similarly, $\AVG \circ \MAX$ guarantees that the average district cost is small, where the cost of a district is now defined as the egalitarian (maximum) cost of any of its members. We consider the introduction and study of these objectives as one of the major contributions of our work.  

We measure the performance of a distributed mechanism by its {\em distortion}, defined as the worst-case ratio (over all instances of the problem) between the objective value of the alternative chosen by the mechanism and the minimum possible objective value achieved over all alternatives. The distortion essentially measures the deterioration of the objective due to the fact that the mechanism must make a decision via a distributed two-step process, on top of other possible informational limitations related to the preferences of the agents. We consider {\em deterministic} mechanisms that are either {\em cardinal} (in which case they have access to the exact distances between agents and alternatives), or {\em ordinal} (in which case they have access only to the rankings that are induced by the distances). Table~\ref{table:results} gives an overview of our bounds on the distortion of distributed mechanisms, for the four objectives defined above. We provide bounds that hold for general metric spaces, and then more refined bounds for the fundamental special case where the metric is a line.

\begin{table}[htb]
\centering
\begin{tabular}{c||c|c|c|c|}
\cline{2-5}
                              & \multicolumn{2}{c|}{General metric}        & \multicolumn{2}{c|}{Line metric}           \\ \cline{2-5} 
                              & Cardinal & Ordinal & Cardinal & Ordinal  \\ \hline \hline
\multicolumn{1}{|c||}{$\AVG\circ \AVG$} & $3$*                & $[7$*$, 11]$              & $3$*                & $7$*               \\ \hline 
\multicolumn{1}{|c||}{$\AVG\circ \MAX$} & $3$                & $[2+\sqrt{5},11]$              & $3$                & $[2+\sqrt{5},5]$               \\ \hline
\multicolumn{1}{|c||}{$\MAX\circ \MAX$} & $[1+\sqrt{2},3]$             & $[3,5]$               & $1+\sqrt{2}$           & $3$               \\ \hline
\multicolumn{1}{|c||}{$\MAX\circ \AVG$} & $[1+\sqrt{2},3]$             & $[2+\sqrt{5},5]$            & $1+\sqrt{2}$       & $[2+\sqrt{5},5]$            \\ \hline
\end{tabular}
\caption{An overview of the distortion bounds for the various settings studied in this paper. Each entry consists of an interval showing a lower bound on the distortion of all distributed mechanisms for the corresponding setting, and an upper bound that is achieved by some mechanism; when a single number is presented, the bound is tight. The results marked with a (*) for the line metric and the $\AVG\circ \AVG$ objective (as well as the corresponding lower bounds for general metrics) follow from the work of \citet{FV2020facility}; all other results in the table were not known previously.} 
\label{table:results} 
\end{table}

Several of our bounds for general metric spaces are based on a novel composition technique for designing distributed mechanisms. In particular, we prove a rather general {\em composition theorem}, which appears in many versions throughout our paper, depending on the objective at hand. Roughly speaking, the theorem relates the distortion of a distributed mechanism to the distortion of the centralized voting rules it uses for the local (in-district) and global (over-districts) aggregation steps. In particular, for two such voting rules with distortion bounds $\alpha$ and $\beta$, the distortion of the composed mechanism is at most $\alpha+\beta+\alpha \beta$. This effectively enables us to plug in voting rules with known distortion bounds, and obtain distributed mechanisms with low distortion. The  theorem is also robust in the sense that the $\AVG$ and $\MAX$ objectives can be substituted with more general objectives satisfying specific properties, such as monotonicity and subadditivity; we provide more details on that in Section \ref{sec:generalization}.

To demonstrate the strength of the composition theorem, consider the objective $\AVG \circ \AVG$. The upper bound of $3$ for cardinal mechanisms and general metrics is obtained by using optimal centralized voting rules (with distortion $1$) for both aggregation steps. Similarly, the upper bound of $11$ for ordinal mechanisms follows by using the ordinal \textsc{PluralityMatching} rule of \citet{gkatzelis2020resolving} in both steps of the distributed mechanism; this rule is known to have distortion at most $3$ for general instances, and at most $2$ when all agents are at distance $0$ from their most-preferred alternative (which is the case when the representatives are thought of as agents in the second step of the mechanism).

Even though the composition theorem is evidently a very powerful tool, it comes short of providing tight bounds in some cases. To this end, we design explicit mechanisms with improved distortion guarantees, both for general metrics as well as the fundamental special case where the metric is a line. A compelling highlight of our work is a novel mechanism for objectives of the form $\MAX \circ G$, to which we refer as $\lambda$-\textsc{Acceptable-Rightmost-Leftmost} ($\lambda$-ARL). While this mechanism has the counter-intuitive property of {\em not} being \emph{unanimous} (i.e., there are cases where all agents agree on the best alternative, but the mechanism does not choose this alternative as the winner), it achieves the best possible distortion of $1+\sqrt{2}$ among all distributed mechanisms on the line. In contrast, we prove that unanimous mechanisms cannot achieve distortion better than $3$. To the best of our knowledge, this is the first time that not satisfying unanimity turns out to be a necessary ingredient for achieving the best possible distortion in the metric social choice literature. 

\subsection{Related Work}
The distortion of centralized social choice voting rules has been studied extensively for many different settings. For a comprehensive introduction to the distortion literature, we refer the interested reader to the recent survey of~\citet{survey}.

After the work of~\citet{procaccia2006distortion}, a series of papers adopted their normalized setting, where the agents have unit-sum values for the alternatives, and proved asymptotically tight bounds on the distortion of ordinal single-winner rules~\citep{caragiannis2011embedding,boutilier2015optimal}, multi-winner rules~\citep{caragiannis2017subset}, rules that choose rankings of alternatives~\citep{benade2019rankings}, and strategyproof rules~\citep{bhaskar2018truthful}. 
Recent papers considered more general questions related to how the distortion is affected by the amount of available information about the values of the agents~\citep{mandal2019thrifty,mandal2020optimal,amanatidis2020peeking}. The normalized distortion has also been investigated in other related problems, such as participatory budgeting~\citep{benade2017participatory}, and one-sided matching~\citep{filos-Ratsikas2014matching,amanatidis2021matching}.

The metric distortion setting was first considered by \citet{anshelevich2018approximating} who, among many results, showed a lower bound of $3$ on the distortion of deterministic single-winner ordinal rules for the social cost, and an upper bound of $5$, achieved by the Copeland rule. Following their work, many papers were devoted to bridging this gap (e.g., see \citep{munagala2019improved,kempe2020duality}) until, finally, \citet{gkatzelis2020resolving} designed the {\sc PluralityMatching} rule that achieves an upper bound of $3$; in fact, this bound holds for the more general {\em fairness ratio}~\citep{goel2017metric} (which captures various different objectives, including the social cost and the maximum cost). Besides the main setting, many other works have shown bounds on the metric distortion for randomized rules~\citep{anshelevich2017randomized,feldman2016voting}, rules that use less than ordinal information about the preferences of the agents~\citep{fain2019random,kempe2020communication,anagnostides2021metric}, committee elections~\citep{chen2020favorite,jaworski20evaluating}, primary elections~\citep{borodin2019primaries}, 
and for many other problems~\citep{abramowitz2017utilitarians,anshelevich2018ordinal}.

Most related to our work are the recent papers of~\citet{FMV2020distributed} and~\citet{FV2020facility}, who initiated the study of the distortion in distributed normalized and metric social choice settings, respectively. As already previously discussed, we improve the results of \citet{FV2020facility} by extending them to hold for general metrics and asymmetric districts, and also show bounds for many other objectives. In our terminology, \citeauthor{FV2020facility} showed a tight bound of $3$ for cardinal distributed mechanisms and a tight distortion of $7$ for ordinal mechanisms, when the metric is a line, the districts have the same size, and the objective is the social cost (which is equivalent to our $\AVG \circ \AVG$ objective when the districts are symmetric). Interestingly, not only do we generalize these results to hold for asymmetric districts and other objectives, but our composition theorem also provides easier proofs, compared to the characterizations of worst-case instances used in their paper.


\section{Preliminaries} \label{sec:prelim}
An instance of our problem is defined as a tuple $I=(N, A, D, \delta)$, where
\begin{itemize}
\item $N$ is a set of $n$ {\em agents}.

\item $A$ is a set of $m$ {\em alternatives}.

\item $D$ is a collection of $k$ {\em districts}, which define a partition of $N$ (i.e., each agent belongs to a single district). Let $N_d$ be the set of agents that belong to district $d \in D$, and denote by $n_d = |N_d|$ the size of $d$. 

\item $\delta$ is a {\em metric space} that contains points representing the agents and the alternatives. In particular, $\delta$ defines a {\em distance} $\delta(i,j)$ between any $i, j \in N \cup A$, such that the {\em triangle inequality} is satisfied, i.e., $\delta(i,j) \leq \delta(i,x) + \delta(x,j)$ for every $i,j,x \in N \cup A$. 
\end{itemize}
A {\em distributed} mechanism takes as input information about the metric space, which can be of cardinal or ordinal nature (e.g., agents and alternatives could specify their exact distances between them, or the linear orderings that are induced by the distances), and outputs a single winner alternative $w \in A$ by implementing the following two steps:
\begin{itemize}
\item Step 1: For every district $d \in D$, the agents therein decide a {\em representative} alternative $y_d \in A$.
\item Step 2: Given the district representatives, the output is an alternative $w \in A$.
\end{itemize}
In both steps, the decisions are made by using {\em direct voting rules}, which map the preferences of a given subset of agents to an alternative. To be more specific, in the first step, an {\em in-district} direct voting rule is applied for each district $d \in D$ with input the preferences of the agents in the district (set $N_d$) to decide its representative $y_d \in A$. Then, in the second step, the district representatives can be thought of as pseudo-agents, and an {\em over-districts} direct voting rule is applied with input their preferences to decide the final winner $w \in A$. In the special case of instances consisting of a single district, the process is not distributed, and thus the two steps collapse into one: The final winner is the alternative chosen to be the district's representative.  

\subsection{Objectives}
We consider standard minimization objectives that have been studied in the related literature, and also propose new ones that are appropriate in the context distributed setting. Each objective assigns a value to every alternative as a cost function composition $F \circ G$ of an objective function $F$ that is applied over the districts and an objective function $G$ that is applied within the districts. Our main four objectives are defined by considering all possible combinations of $F, G \in \{\AVG,\MAX\}$, where the functions $\AVG$ and $\MAX$ define an average and a max over districts or agents within a district, respectively. In particular, we have:
\begin{itemize}
\item 
The $\AVG \circ \AVG$ of an alternative $j \in A$ is defined as 
$$(\AVG \circ \AVG)(j | I) = \frac{1}{k} \sum_{d \in D} \bigg( \frac{1}{n_d} \sum_{i \in N_d} \delta(i,j) \bigg).$$

\item 
The $\MAX \circ \MAX$ of an alternative $j \in A$ is defined as 
$$(\MAX \circ \MAX)(j | I) = \max_{d \in D} \max_{i \in N_d} \delta(i,j) = \max_{i \in N} \delta(i,j).$$

\item 
The $\AVG \circ \MAX$ of an alternative $j \in A$ is defined as 
$$(\AVG \circ \MAX)(j | I) = \frac{1}{k} \sum_{d \in D} \max_{i \in N_d} \delta(i,j).$$

\item
The $\MAX \circ \AVG$ of an alternative $j \in A$ is defined as 
$$(\MAX \circ \AVG)(j | I) = \max_{d \in D} \bigg\{ \frac{1}{n_d} \sum_{i \in N_d} \delta(i,j) \bigg\}.$$
\end{itemize}
The $\AVG \circ \AVG$ objective is similar to the well-known utilitarian {\em average social cost} objective measuring the average total distance between all agents and alternative $j$; actually, $\AVG \circ \AVG$ coincides with the average social cost when the districts are symmetric (i.e., have the same size), but not in general. The $\MAX \circ \MAX$ objective coincides with the egalitarian {\em max cost} measuring the maximum distance from $j$ among all agents. The new objectives $\AVG \circ \MAX$ and $\MAX \circ \AVG$ make sense in the context of distributed voting, and can be thought of as measures of fairness between districts. For example, minimizing the $\MAX \circ \AVG$ objective corresponds to making sure that the final choice treats each district fairly so that the average social cost of each district is almost equal to that of any other district. Of course, besides combinations of $\AVG$ and $\MAX$, one can define many more objectives; we consider such generalizations in Section~\ref{sec:generalization}.

\subsection{Distortion of voting rules and distributed mechanisms}
Direct voting rules can be suboptimal, especially when they have limited access to the metric space (for example, when ordinal information is known about the preferences of the agents over the alternatives). This inefficiency is typically captured in the related literature by the notion of distortion, which is the worst-case ratio between the objective value of the optimal alternative over the objective value of the alternative chosen by the rules. Formally, given a minimization cost objective $F \in \{\AVG,\MAX\}$, the {\em $F$-distortion} of a voting rule $V$ is
\begin{align*}
\dist_F(V) = \sup_{I=(N,A,\delta)} \frac{F(V(I)|I)}{\min_{j \in A} F(j|I)},    
\end{align*}
where $V(I)$ denotes the alternative chosen by the voting rule when given as input the (single-district) instance $I$ consisting of a set of agents $N$, a set of alternatives $A$, and a metric space $\delta$.

The notion of distortion can be naturally extended for the case of distributed mechanisms. Given a composition objective $F\circ G$,  the {\em $(F\circ G)$-distortion} of a distributed mechanism $M$ is 
\begin{align*}
\dist_{F \circ G}(M) = \sup_{I = (N,A,D,\delta)} \frac{(F \circ G)(M(I)|I)}{\min_{j \in A} (F \circ G)(j|I)},    
\end{align*}
where $M(I)$ is the alternative chosen by the mechanism when given as input an instance $I$ consisting of a set of agents $N$, a set of alternatives $A$, a set $D$ of districts, and a metric space $\delta$. Our goal is to bound the distortion of distributed mechanisms for the different objectives we consider. To this end, we will either show how known results from the literature about the distortion of direct voting rules can be composed to yield distortion bounds for distributed mechanisms, or design explicit mechanisms with low distortion.


\section{Composition Results for General Metric Spaces} \label{sec:metric}
In this section we consider general metric spaces, and show how known distortion bounds for direct voting rules can be composed to yield distortion bounds for distributed mechanisms that rely on those voting rules. Given an objective $F \circ G$, we say that a distributed mechanism is {\em $\alpha$-in-$\beta$-over} if it works as follows. \smallskip

\begin{tcolorbox}[title={$\alpha$-in-$\beta$-over mechanisms},colback=white]
\begin{enumerate}[left=-1pt]
\item For each district $d \in D$, choose its representative using an in-district voting rule with $G$-distortion at most $\alpha$.
\item Choose the final winner using an over-districts voting rule with $F$-distortion at most $\beta$.
\end{enumerate}
\end{tcolorbox}

Our first technical result is an upper bound on the $(F \circ G)$-distortion of $\alpha$-in-$\beta$-over mechanisms, for any $F, G \in \{\AVG, \MAX\}$. The proof of the following theorem also follows from the more general Theorem~\ref{thm:general-composition} in Section~\ref{sec:generalization1}, which considers objectives that are compositions of functions satisfying particular properties, such as monotonicity and subadditivity. 

\begin{theorem} \label{thm:main-composition}
For any $F, G \in \{\AVG, \MAX\}$, the $(F \circ G)$-distortion of any $\alpha$-in-$\beta$-over mechanism is at most $\alpha+\beta+\alpha\beta$.
\end{theorem}

\begin{proof}
Here, we present a proof only for the $\AVG \circ \AVG$ objective; the proof for the other objectives is similar. Consider an arbitrary $\alpha$-in-$\beta$-over mechanism $M$ and an arbitrary instance $I=(N, A, D, \delta)$. Let $w$ be the alternative that $M$ outputs as the final winner when given $I$ as input, and denote by $o$ an optimal alternative. By the definition of $M$, we have the following two properties:
\begin{align}\label{eq:SS-alpha-in-district}
 \forall j \in A, d \in D: \sum_{i \in N_d} \delta(i, y_d) \leq \alpha \sum_{i \in N_d} \delta(i, j)
\end{align}
and
\begin{align}\label{eq:SS-beta-over-districts}
 \forall j \in A: \sum_{d \in D} \delta(y_d, w) \leq \beta \sum_{d \in D} \delta(y_d, j)
\end{align}
By the triangle inequality, we have $\delta(i,w) \leq \delta(y_d, w) + \delta(i, y_d)$ for any agent $i \in N$.
Using this, 
we obtain
\begin{align*}
(\AVG \circ \AVG)(w|I) 
&= \frac{1}{k} \sum_{d \in D} \left( \frac{1}{n_d} \sum_{i \in N_d} \delta(i,w) \right) \\
&\leq \frac{1}{k} \sum_{d \in D} \left( \frac{1}{n_d} \sum_{i \in N_d} \delta(y_d, w) \right) + \frac{1}{k} \sum_{d \in D} \left( \frac{1}{n_d} \sum_{i \in N_d} \delta(i, y_d) \right) \\
&= \frac{1}{k} \sum_{d \in D} \delta(y_d, w) + \frac{1}{k} \sum_{d \in D} \left( \frac{1}{n_d} \sum_{i \in N_d} \delta(i, y_d) \right).
\end{align*}
By \eqref{eq:SS-alpha-in-district} and \eqref{eq:SS-beta-over-districts} for $j=o$, we obtain
\begin{align*}
(\AVG \circ \AVG)(w|I)
&\leq \beta \cdot \frac{1}{k} \sum_{d \in D} \delta(y_d,o)   + \alpha \cdot \frac{1}{k} \sum_{d \in D} \left( \frac{1}{n_d} \sum_{i \in N_d} \delta(i,o) \right) \nonumber \\
&= \beta \cdot \frac{1}{k}\sum_{d \in D} \left( \frac{1}{n_d} \sum_{i \in N_d} \delta(y_d,o) \right)  + \alpha \cdot (\AVG \circ \AVG)(o|I).
\end{align*}
By the triangle inequality, we have $\delta(y_d,o)  \leq \delta(i,y_d) + \delta(i,o)$ for any agent $i \in N$.
Using this and \eqref{eq:SS-alpha-in-district} for $j=o$, we can upper bound the first term of the last expression above as follows:
\begin{align*}
\beta \cdot \frac{1}{k} \sum_{d \in D} \left( \frac{1}{n_d} \sum_{i \in N_d} \delta(y_d,o) \right)
&\leq \beta \cdot \frac{1}{k} \sum_{d \in D} \left( \frac{1}{n_d} \sum_{i \in N_d} \delta(i,y_d) \right) + \beta \cdot \frac{1}{k} \sum_{d \in D} \left( \frac{1}{n_d} \sum_{i \in N_d} \delta(i,o) \right) \\
&\leq(\beta+\alpha \beta) \cdot \frac{1}{k} \sum_{d \in D} \left( \frac{1}{n_d} \sum_{i \in N_d} \delta(i,o) \right) \\
&= (\beta+\alpha \beta) \cdot (\AVG \circ \AVG)(o|I).
\end{align*}
Putting everything together, we obtain a distortion upper bound of $\alpha+\beta+\alpha\beta$.
\end{proof}

By applying Theorem~\ref{thm:main-composition} using known results, we can derive bounds on the $(F \circ G)$-distortion of distributed mechanisms, for any $F, G \in \{\AVG, \MAX\}$. In particular, due to the structure of $F$ and $G$, if the whole metric space is known (that is, we have access to the exact distances between agents and alternatives), then we can easily compute the alternative that optimizes $F$ and $G$. In other words, there exist direct voting rules with $F$- and $G$-distortion $1$, which can be used to obtain a $1$-in-$1$-over distributed mechanism, and the following statement.

\begin{corollary}\label{cor:main-cardinal}
For any $F, G \in \{\AVG, \MAX\}$, there exists a cardinal distributed mechanism with $(F \circ G)$-distortion at most $3$.
\end{corollary}

If only ordinal information is available about the distances between agents and alternatives, then we can employ the {\sc PluralityMatching} voting rule of \citet{gkatzelis2020resolving} both within and over the districts. This rule is known to achieve the best possible distortion of $3$ among all ordinal rules, for any $F, G \in \{\AVG, \MAX\}$. In fact, this rule achieves a distortion bound of $2$ when all agents are at distance $0$ from their top alternative; this is the case when the agents are a subset of the alternatives as in the second step of a distributed mechanism.\footnote{See Theorem 1 and Proposition 6 in the arxiv version of the paper of \citet{gkatzelis2020resolving}.} Hence, we have a $3$-in-$2$-over mechanism, and Theorem~\ref{thm:main-composition} yields the following statement. 

\begin{corollary}\label{cor:main-ordinal}
For any $F, G \in \{\AVG, \MAX\}$, there exists an ordinal distributed mechanism with $(F \circ G)$-distortion at most $11$.
\end{corollary}

Corollaries~\ref{cor:main-cardinal} and~\ref{cor:main-ordinal} demonstrate the power that Theorem~\ref{thm:main-composition} gives us in designing distributed mechanisms with constant distortion upper bounds, by using known results from the literature as black boxes. However, this method does not always lead to the best possible distributed mechanisms. In particular, let us consider the objectives $\MAX \circ G$, for $G \in \{\AVG,\MAX\}$ and the class of ordinal mechanisms. We can improve upon the upper bound of $11$ due to Corollary~\ref{cor:main-ordinal} using the following, much simpler mechanisms. \\

\begin{tcolorbox}[title={$\alpha$-in-arbitrary-over mechanisms},colback=white]
\begin{enumerate}[left=-1pt]
\item For each district $d \in D$, choose its representative using an ordinal in-district voting rule with $G$-distortion at most $\alpha$.
\item Output an arbitrary representative as the final winner.
\end{enumerate}
\end{tcolorbox}

\begin{theorem} \label{thm:alpha-in-arbitrary-over}
For any $G \in \{\AVG,\MAX\}$, the $(\MAX \circ G)$-distortion of any $\alpha$-in-arbitrary-over mechanism is at most $2+\alpha$. 
\end{theorem}

\begin{proof}
We will present a proof only for the $\MAX \circ \AVG$ objective; the proof for $\MAX \circ \MAX$ follows by similar, even simpler arguments. Consider any $\alpha$-in-arbitrary-over mechanism $M$ and any instance $I=(N,A,D,\delta)$. Let $w$ be the alternative chosen by $M$ when given as input the ordinal information of $I$, and denote the optimal alternative by $o$. In addition, let $d^* \in \arg\max_{d \in D} \frac{1}{n_d} \sum_{i \in N_d} \delta(i,w)$ be the district that gives the max cost for $w$. By the triangle inequality, and since 
$(\MAX \circ \AVG)(o|I) \geq \frac{1}{n_d} \sum_{i \in N_d} \delta(i,o)$ for every $d \in D$, we have
\begin{align*}
(\MAX \circ \AVG)(w|I) 
&= \frac{1}{n_{d^*}} \sum_{i \in N_{d^*}} \delta(i,w) \\
&\leq \frac{1}{n_{d^*}} \sum_{i \in N_{d^*}} \bigg( \delta(i,o) +  \delta(w,o) \bigg) \\
&\leq (\MAX \circ \AVG)(o|I) + \delta(w,o)
\end{align*}
Now, let $d_w$ be the district whose representative is $w$. By the triangle inequality and the fact that the in-district voting rule used to choose $w$ as the representative of $d_w$ has $\AVG$-distortion at most $\alpha$, we obtain
\begin{align*}
\delta(w,o) 
&\leq \frac{1}{n_{d_w}} \sum_{i \in N_{d_w}} \bigg( \delta(i,o) + \delta(i,w) \bigg) \\
&\leq (1+\alpha) \cdot \frac{1}{n_{d_w}} \sum_{i \in N_{d_w}} \delta(i,o) \\
&\leq (1+\alpha) \cdot (\MAX \circ \AVG)(o|I).
\end{align*}
Putting everything together, we obtain a distortion upper bound of $2+\alpha$.
\end{proof}

Now, using again the voting rule of \citet{gkatzelis2020resolving} with $\AVG$- and $\MAX$-distortion at most $3$ as an in-district rule, we obtain a $3$-in-arbitrary-over distributed mechanism, and Theorem~\ref{thm:alpha-in-arbitrary-over} implies the following result.

\begin{corollary}\label{cor:5-max-of-G-ordinal}
For any $G \in \{\AVG,\MAX\}$, there exists an ordinal distributed mechanism with $(\MAX \circ G)$-distortion at most $5$.
\end{corollary}


\section{Improved Results on the Line Metric} \label{sec:line}

In this section we focus  on the line metric, where both the agents and the alternatives are assumed to be points on the line of real numbers. Exploiting this structure, there are classes of mechanisms for which we can obtain significantly improved bounds compared to those implied by the general composition Theorem~\ref{thm:main-composition}, as well as Theorem~\ref{thm:alpha-in-arbitrary-over}. 

\subsection{Ordinal Mechanisms}
We start with ordinal distributed mechanisms and the two objectives $\AVG \circ G$ for $G \in \{\AVG, \MAX\}$. Recall that Corollary~\ref{cor:main-ordinal} implies a distortion bound of at most $11$ for these objectives. However, when the metric is a line, we can do much better by observing that there is an ordinal direct over-districts voting rule with $\AVG$-distortion of $1$. In particular, we can identify the median district representative and choose it as the final winner. Using the rule of \citet{gkatzelis2020resolving} as the direct in-district voting rule, we obtain a distributed $3$-in-$1$-over mechanism with distortion at most $7$ due to Theorem~\ref{thm:main-composition}. 

\begin{corollary}\label{cor:SUM-G-line}
When the metric is a line, there exists an ordinal distributed mechanism with $(\AVG \circ G)$-distortion at most $7$, for any $G \in \{\AVG, \MAX\}$. 
\end{corollary}

Corollary~\ref{cor:SUM-G-line} essentially recovers the tight distortion bound of $7$ by \citet{FV2020facility} for $\AVG \circ \AVG$ when the districts are symmetric, and also extends it to the case of asymmetric districts. For $\AVG \circ \MAX$, this bound of $7$ is a first improvement, but we can do even better with the following distributed mechanism. \\

\begin{tcolorbox}[title={\sc Arbitrary-Median},colback=white]
\begin{enumerate}[left=-1pt]
    \item For every district $d \in D$, choose its representative to be the favorite alternative of an arbitrary agent $j_d \in N_d$.
    \item Output the median representative as the final winner.
\end{enumerate}
\end{tcolorbox}

\begin{theorem} \label{thm:arbitrary-median}
When the metric is a line, {\sc Arbitrary-Median} has $(\AVG \circ \MAX)$-distortion at most $5$.
\end{theorem}

\begin{proof}
Consider any instance $I=(N,A,D,\delta)$, in which $\delta$ is a line metric. Let $w$ be the alternative chosen by {\sc Arbitrary-Median} when given as input (the ordinal information of) $I$, and denote by $o$ the optimal alternative. For every district $d \in D$, let $i_d = \arg\max_{i \in N_d} \delta(i,w)$ be the agent in $d$ that gives the max cost for $w$. By the triangle inequality, and since $(\AVG \circ \MAX)(o|I) \geq \frac{1}{k} \sum_{d \in D}\delta(i_d,o)$, we have
\begin{align*}
(\AVG \circ \MAX)(w|I) 
&= \frac{1}{k} \sum_{d \in D} \delta(i_d,w) \\
&\leq \frac{1}{k} \sum_{d \in D} \bigg( \delta(i_d,o) + \delta(w,o) \bigg) \\
&\leq (\AVG \circ \MAX)(o|I) + \delta(w,o). 
\end{align*}
Without loss of generality, we can assume that $w$ is to the left of $o$ on the line. Since $w$ is the median representative, there is a set $S$ of at least $k/2$ districts whose representatives are to the left of (or coincide with) $w$. As a result, for every district $d \in S$, agent $j_d$ (whose favorite alternative becomes the representative of $d$) prefers $w$ over $o$, and thus $d(j_d,o) \geq \delta(w,o)/2$. Using this, we obtain
\begin{align*}
&(\AVG \circ \MAX)(o|I) = \frac{1}{k} \sum_{d \in D} \max_{i \in N_d} \delta(i,o) 
\geq \frac{1}{k} \sum_{d \in S} \delta(j_d,o)
\geq \frac{1}{2} \cdot \frac{\delta(w,o)}{2} \\
&\Leftrightarrow
\delta(w,o) \leq 4 \cdot (\AVG \circ \MAX)(o|I).
\end{align*}
Putting everything together, we obtain a distortion upper bound of $5$. 
\end{proof}

Next, we show an almost matching lower bound of approximately $4.23$. Before going through with the proof, we argue that ordinal distributed mechanisms with finite distortion must satisfy a unanimity property, which will be used extensively in all our lower bound constructions. Formally, we say that a distributed mechanism is {\em unanimous} if it chooses the representative of a district $d$ to be an alternative $a$ whenever all agents in $N_d$ prefer $a$ over all other alternatives.

\begin{lemma} \label{lem:ordinal-unanimous}
For any $F, G \in \{\AVG,\MAX\}$, every ordinal distributed mechanism with finite $(F \circ G)$-distortion must be unanimous.
\end{lemma}

\begin{proof}
Suppose towards a contradiction that there is an ordinal distributed mechanism $M$ with finite $(F \circ G)$-distortion that is not unanimous. Consider an instance $I$ with a single district consisting of agents all of whom prefer alternative $a$ to all other alternatives. Since $M$ is not unanimous, it chooses some different alternative $b$ to be the representative of the district, and thus the overall winner (as there is only a single district). However, if $a$ and all agents are positioned at the same point in the metric space, we have that $(F\circ G)(a|I)=0$ and $(F\circ G)(b|I) > 0$, and the distortion of $M$ is unbounded, a contradiction.
\end{proof}

We are now ready to present the lower bound on the $(\AVG \circ \MAX)$-distortion of ordinal distributed mechanisms.

\begin{theorem} \label{thm:avg-of-max-lower}
The $(\AVG \circ \MAX)$-distortion of any ordinal distributed mechanism is at least $2+\sqrt{5}-\varepsilon$, for any $\varepsilon > 0$, even when the metric is a line.
\end{theorem}

\begin{proof}
Suppose towards a contradiction that there is an ordinal distributed mechanism $M$ with distortion strictly smaller than $2+\sqrt{5}-\varepsilon$, for any $\varepsilon > 0$. We will reach a contradiction by defining instances with two alternatives $a$ and $b$, and districts consisting of the same size. Without loss of generality, we assume that $M$ chooses alternative $a$ as the final winner when given as input any instance with only two districts, such that both alternatives are representative of some district. 
Let $x$ and $y$ be two integers such that $\phi > y/x \ge\phi - \varepsilon/2$, where $\phi = (1+\sqrt{5})/2$ is the golden ratio.

First, consider an instance $I_1$ consisting of the following two districts:
\begin{itemize}
\item The first district consists of two agents, such that one of them prefers alternative $a$ and the other prefers alternative $b$. 
\item The second district consists of two agents, such that both of them prefer alternative $b$. Due to unanimity (Lemma~\ref{lem:ordinal-unanimous}), the representative of this district must be $b$.  
\end{itemize}
Suppose that $M$ chooses $a$ as the representative of the first district, in which case both alternatives are representative of some district, and thus $M$ chooses $a$ as the final winner. Consider the following metric:
\begin{itemize}
\item Alternative $a$ is positioned at $0$, and alternative $b$ is positioned at $1$.
\item In the first district, the agent that prefers alternative $a$ is positioned at $1/2$, whereas the agent that prefers alternative $b$ is positioned at $3/2$. 
\item In the second district, both agents are positioned at $1$.
\end{itemize}
Then, we have that $(\AVG \circ \MAX)(a|I_1) = \frac{1}{2}\left( \frac{3}{2} + 1 \right) = 5/4$ and $(\AVG \circ \MAX)(b|I_2) = \frac{1}{2} \left(\frac{1}{2}+0\right) = 1/4$, leading to a distortion of $5$. Consequently, $M$ must choose $b$ as the representative of the first district (where one agent prefers $a$ and one prefers $b$). 

Next, we argue that for instances with $x+y$ districts such that $a$ is the representative of $x$ districts and $b$ is the representative of $y$ districts, $M$ must choose $b$ as the final winner. Assume otherwise that $M$ chooses $a$ in such a situation, and consider an instance $I_2$ consisting of the following $x+y$ districts:
\begin{itemize}
\item Each of the first $x$ districts consists of a single agent that prefers alternative $a$; thus, $a$ is the representative of all these districts.
\item Each of the next $y$ districts consists of a single agent that prefers alternative $b$; thus, $b$ is the representative of all these districts.
\end{itemize}
By assumption, $M$ chooses $a$ as the final winner. Consider the following metric:
\begin{itemize}
\item Alternative $a$ is positioned at $0$ and alternative $b$ is positioned at $1$.
\item In each of the first $x$ districts, the agent therein is positioned at $1/2$.
\item In each of the next $y$ districts, the agent therein is positioned at $1$.
\end{itemize}
Since 
$$(\AVG \circ \MAX)(a|I_2) = \frac{1}{x+y} \left(\frac{x}{2} + y\right) = \frac{1}{x+y} \cdot \frac{x+2y}{2}$$ 
and 
$$(\AVG \circ \MAX)(b|I_2) = \frac{1}{x+y} \left(\frac{x}{2} + 0 \right) = \frac{1}{x+y} \cdot \frac{x}{2},$$ 
the distortion is 
$$\frac{x+2y}{x} = 1 + 2\frac{y}{x} \ge 1+2\phi - \varepsilon = 2+\sqrt{5} - \varepsilon.$$ 
Consequently, $M$ must choose $b$, whenever there are $x+y$ districts such that $a$ is the representative of $x$ districts and $b$ is the representative of the remaining $y$ districts.

Finally, consider an instance $I_3$ with the following $x+y$ districts:
\begin{itemize}
\item Each of the first $x$ districts consists of two agents that prefer alternative $a$. Due to unanimity, $a$ must be the representative of all these districts.
\item Each of the next $y$ districts consists of two agents, such that one of them prefers alternative $a$, while the other prefers alternative $b$. By the discussion above (about instance $I_1$), $b$ must be the representative of these districts.
\end{itemize}
Since $a$ is the representative of $x$ districts and $b$ is the representative of $y$ districts, by the discussion above (about instance $I_2$), $M$ chooses $b$ as the final winner. Consider the following metric:
\begin{itemize}
\item Alternative $a$ is positioned at $0$ and alternative $b$ is positioned at $1$.
\item In each of the first $x$ districts, the two agents therein are positioned at $0$.
\item In each of the next $y$ districts, the agent that prefers $a$ is positioned at $-1/2$, whereas the agent that prefers $b$ is positioned at $1/2$.
\end{itemize}
Since 
$$(\AVG \circ \MAX)(a|I_3) = \frac{1}{x+y} \left( 0 + \frac{y}{2} \right) = \frac{1}{x+y} \cdot \frac{y}{2}$$
and 
$$(\AVG \circ \MAX)(b|I_2) = \frac{1}{x+y} \left( x + \frac{3y}{2}\right) = \frac{1}{x+y} \cdot \frac{2x+3y}{2},$$ 
the distortion is 
$$\frac{2x+3y}{y} = 3 + \frac{2x}{y} 
> 3 + \frac{2}{\phi} =  1+2\phi = 2+\sqrt{5}.$$
This contradicts our assumption that $M$ has distortion strictly smaller than $2+\sqrt{5}-\varepsilon$, thus completing the proof.
\end{proof}


Next, we consider the objectives $\MAX \circ G$ for $G \in \{\AVG, \MAX\}$. Corollary~\ref{cor:5-max-of-G-ordinal} implies a distortion bound of at most $5$ for these objectives. When $G = \MAX$ and the metric is a line, we can get an improved bound of $3$ using a rather simple distributed mechanism, which essentially outputs the favorite alternative of an arbitrary agent. \\

\begin{tcolorbox}[title={\sc Arbitrary-Dictator},colback=white]
\begin{enumerate}[left=-1pt]
\item For each district $d \in D$, choose its representative to be the favorite alternative of an arbitrary agent in $N_d$.
\item Output an arbitrary district representative as the final winner. 
\end{enumerate}
\end{tcolorbox}

\begin{theorem} \label{thm:arbitrary-dictator}
When the metric is a line, {\sc Arbitrary-Dictator} has $(\MAX \circ \MAX)$-distortion at most $3$.
\end{theorem}

\begin{proof}
Let $I=(N, A, D, \delta)$ be an arbitrary instance, with $\delta$ a line metric. Without loss of generality, we assume that the alternative $w$ chosen by {\sc Arbitrary-Dictator} is positioned to the right of the dictator agent $i^*$. In addition, we denote by $o$ the optimal alternative, by $\ell$ the leftmost agent, and by $r$ the rightmost agent. With some abuse of notation, in the following we will also use $w$, $i^*$, $o$, $\ell$ and $r$ to denote the positions of the corresponding agents and alternatives on the line. 

First observe that the distortion is at most $2$ when $w \in [\ell, r]$. In particular, since there is at least one alternative in-between the agents, we have that $(\MAX \circ \MAX)(o | I) \geq \frac{\delta(\ell,r)}{2}$ and $(\MAX \circ \MAX)(w |I) \leq \delta(\ell,r)$. Hence, we can assume that $w > r$ (since $w \geq i^*$). Furthermore, it cannot be the case that $i^* \leq o \leq w$ (since then $i^*$ would prefer $o$ over $w$), neither that $r \leq w \leq o$ (since then $o$ would not be the optimal alternative). Consequently, $o < i^* \leq r < w$. Since $w$ is the favorite alternative of $i^*$, $w$ is also the favorite alternative of $r$. Hence, $\delta(r,w) \leq \delta(r,o)$. 
We distinguish between the following two subcases:
\begin{itemize}
\item
If $o \in [\ell,r]$, we clearly have that $(\MAX \circ \MAX)(o | I) \geq \frac{\delta(\ell,r)}{2}$, and 
$$(\MAX \circ \MAX)(w | I) = \delta(\ell,w) = \delta(\ell,r) + \delta(r,w) \leq 3 \cdot (\MAX \circ \MAX)(o | I).$$

\item 
If $o < \ell$, we have that 
$$(\MAX \circ \MAX)(o|I) = \delta(r,o) = \delta(\ell,o) + \delta(\ell,r) \geq \delta(\ell,r)$$
and
$$(\MAX \circ \MAX)(w|I) = \delta(\ell,w) = \delta(r,w) + \delta(\ell,r) \leq \delta(r,o) + \delta(\ell,r) \leq 2 \cdot (\MAX \circ \MAX)(o|I).$$
\end{itemize}
This completes the proof.
\end{proof}

The following theorem shows that the bound of $3$ is the best possible we can hope for the $\MAX \circ \MAX$ objective using a unanimous distributed mechanism, even when the metric is a line. This lower bound directly extends to ordinal mechanisms, as any such mechanism with finite distortion has to be unanimous (Lemma~\ref{lem:ordinal-unanimous}).

\begin{theorem} \label{thm:max-of-max-lower}
The $(\MAX \circ \MAX)$-distortion of any unanimous distributed mechanism is at least $3$, even when the metric is a line.
\end{theorem}

\begin{proof}
Consider an arbitrary unanimous distributed mechanism $M$, and the following instance $I$ with two alternatives $a$ and $b$. There are two districts with $x \geq 1$ agents each: All agents in the first district prefer alternative $a$, whereas all agents in the second district prefer $b$.
Due to unanimity, $a$ must be the representative of the first district, and $b$ the representative of the second district. The final winner can be any of the two, say $b$ without loss of generality. Now, consider the following metric:
\begin{itemize}
\item Alternative $a$ is located at $1$ and alternative $b$ is located at $3$.
\item All agents in the first district (who prefer $a$) are located at $0$.
\item All agents in the second district (who prefer $b$) are located at $2$.
\end{itemize}
Consequently, we have that $(\MAX \circ \MAX)(a|I)=1$ and $(\MAX \circ \MAX)(b|I)=3$, leading to a distortion of $3$. 
\end{proof}

Finally, let us focus on the objective $\MAX \circ \AVG$, for which we show a lower bound of approximately $4.23$ on the distortion of all ordinal distributed mechanisms, thus almost matching the upper bound of $5$ implied by Corollary~\ref{cor:5-max-of-G-ordinal}.  

\begin{theorem} \label{thm:max-of-avg-lower}
The $(\MAX \circ \AVG)$-distortion of any ordinal distributed mechanism is at least $2+\sqrt{5}-\varepsilon$, for any $\varepsilon>0$, even when the metric is a line.
\end{theorem}

\begin{proof}
Suppose towards a contradiction that there is an ordinal distributed mechanism $M$ with $(\MAX \circ \AVG)$-distortion strictly smaller than $2+\sqrt{5}-\varepsilon$, for any $\varepsilon>0$. We will consider instances with two alternatives $a$ and $b$. We can assume without loss of generality that $M$ outputs $a$ as the winner whenever it is given as input an instance with two districts of the same size, such that both alternatives are representative of some district. However, we do not know how $M$ decides the representative of a district in case there is a tie between the two alternatives therein. Let $\theta = \frac{\sqrt{5}-1}{2} = \phi-1 \approx 0.618$, and denote by $x$ a sufficiently large integer, such that $\theta x$ and $(1-\theta)x$ are both integers.\footnote{Here, we assume the existence of such an $x$ to simplify the presentation of the proof. Since $\theta$ is irrational, to be exact, we can choose $x$ to be a sufficiently large integer so that $\theta x$ and $(1-\theta)x$ are approximately equal to their floors, which would lead to the distortion lower bound being $2+\sqrt{5}-\varepsilon$, for any $\varepsilon >0$.} Observe that $\theta$ is such that $\frac{1+\theta}{1-\theta} = \frac{2+\theta}{\theta} = 1+2\phi = 2+\sqrt{5}$.

First, consider an instance $I_1$ with a single district consisting of $x$ agents, such that 
$\theta x$ agents prefer $a$, and $(1-\theta) x$ agents prefer $b$. 
Suppose that $M$ outputs $b$ (as the district representative and the final winner) when given these ordinal preferences of the agents in the district as input, and consider the following metric: 
\begin{itemize}
\item Alternative $a$ is located at $0$ and alternative $b$ is located at $1$.
\item The $\theta x$ agents that prefer $a$ are located at $0$.
\item The $(1-\theta) x$ agents that prefer $b$ are located at $1/2$.
\end{itemize}
Since there is only one district, we clearly have that
$$(\MAX \circ \AVG)(a|I_1) = \frac{1}{x}\left( 0 +  \frac{(1-\theta) x}{2} \right) = \frac{1}{x} \cdot  \frac{(1-\theta)x}{2}$$
and
$$(\MAX \circ \AVG)(b|I_1) = \frac{1}{x}\left( \theta x +  \frac{(1-\theta) x}{2} \right) = \frac{1}{x} \cdot  \frac{(1+\theta)x}{2}.$$
Hence, the distortion in this case is $\frac{1+\theta}{1-\theta} = 2+\sqrt{5}$. As this contradicts our assumption that $M$ has distortion strictly smaller than $2+\sqrt{5}$, $M$ must choose $a$ as the winner in such a single-district instance. 

Now, consider an instance $I_2$ with the following two districts:
\begin{itemize}
\item The first district is similar to the one in instance $I_1$, i.e., it consists of $\theta x$ agents that prefer $a$, and $(1-\theta) x$ agents that prefer $b$. By the above discussion, the representative of this district must be  alternative $a$. 

\item The second district consists of $x$ agents, all of whom prefer $b$. Due to unanimity (Lemma~\ref{lem:ordinal-unanimous}), the representative of this district must be alternaive $b$. 
\end{itemize}
Since both $a$ and $b$ are representative of some district, $M$ outputs $a$ as the final winner.
Now, consider the following metric:
\begin{itemize}
\item Alternative $a$ is located at $0$ and alternative $b$ is located at $1$. 

\item In the first district, the $\theta x$ agents that prefer $a$ are located at $1/2$, and the remaining $(1-\theta) x$ agents that prefer $b$ are located at $1$. 
The total distance of the agents in the district is 
$\frac{\theta x}{2} + (1-\theta) x = \frac{(2-\theta)x}{2}$ from $a$,
and 
$\frac{\theta x}{2}$ from $b$. 

\item In the second district, all $x$ agents are located at $1 + \frac{\theta}{2}$. The total distance of the agents in the district is $\frac{(2 + \theta)x}{2}$ from $a$, and $\frac{\theta x}{2}$ from $b$. 
\end{itemize}
Consequently, from the second district, we have that 
$$(\MAX \circ \AVG)(a|I_2) = \frac{1}{x} \cdot \frac{(2 + \theta)x}{2}$$
and, from both districts, we have that
$$(\MAX \circ \AVG)(b|I_2) = \frac{1}{x} \cdot \frac{\theta x}{2}.$$
Thus, the distortion of the mechanism for this instance is $\frac{2+\theta}{\theta} = 2 + \sqrt{5}$, a contradiction. 
\end{proof}

\subsection{Cardinal Mechanisms} \label{sec:line-cardinal}
We now turn out attention to distributed mechanisms that have access to the line metric, and are thus aware of the distances between agents and alternatives. Recall that for such mechanisms, Corollary~\ref{cor:main-cardinal} implies a distortion bound of $3$ for all the objectives we have considered so far. As in the case of ordinal mechanisms, when the metric is a line, \citet{FV2020facility} showed a matching lower bound of $3$ for $\AVG \circ \AVG$ (when the districts are symmetric), which extends for $\AVG \circ \MAX$ as the construction also works for instances with  single-agent districts, in which case $\MAX = \AVG$. 

For objectives of the form $\MAX \circ G$, we design a novel distributed mechanism that is tailor-made for the line metric and achieves a distortion of at most $1+\sqrt{2} \leq 2.42$. This mechanism is particularly interesting as it is {\em not} unanimous: Even when all agents in a district prefer an alternative $a$ to every other alternative (i.e., $a$ is the closest alternative to all agents), the mechanism may end up choosing a different alternative as the representative. In fact, to break the distortion barrier of $3$, our mechanism {\em has} to be non-unanimous; by setting $x=1$ in the proof of Theorem~\ref{thm:max-of-max-lower}, we have that any unanimous distributed mechanism cannot achieve a $(\MAX \circ G)$-distortion better than 3. 

For a given $\lambda \geq 1$, we say that an alternative is {\em $\lambda$-acceptable} for a district $d\in D$ if her $G$-value for the agents in $N_d$ is at most $\lambda$ times the $G$-value of any other alternative for the agents in $N_d$. Given an objective $G$, we define a class of distributed mechanisms parameterized by $\lambda$ that work as follows: \\

\begin{tcolorbox}[title={\sc$\lambda$-Acceptable-Rightmost-Leftmost ($\lambda$-ARL)},colback=white]
\begin{enumerate}[left=-1pt]
\item For each district $d$, choose its representative to be the \emph{rightmost} $\lambda$-acceptable alternative for the district.
\item Output the leftmost district representative as the final winner.
\end{enumerate}
\end{tcolorbox}

\begin{theorem}\label{thm:max-of-g-line-alpha}
For any $G \in \{\AVG,\MAX\}$, the $(\MAX \circ G)$-distortion of $\lambda$-ARL is at most $\max\left\{2+\frac{1}{\lambda},\lambda\right\}$.
\end{theorem}

\begin{proof}
We will present the proof only for the $\MAX \circ \AVG$ objective; the proof for $\MAX \circ \MAX$ follows by similar, even simpler arguments. Consider an arbitrary instance $I = (N,A,D,\delta)$, where $\delta$ is a line metric. 
Let $w$ be the alternative selected by $\lambda$-ARL when given as input $I$, and denote by $o$ the optimal alternative. We consider the following two cases depending on the relative positions of $w$ and $o$.

\medskip

\noindent 
{\bf Case 1: $w$ is to the left of $o$.} 
Let $d^* \in \arg\max_{d \in D} \sum_{i \in N_d} \delta(i,w)$ be the district that gives the max cost for $w$. By the triangle inequality, and since $(\MAX \circ \AVG)(o|I) \geq \frac{1}{n_d} \sum_{i \in N_d} \delta(i,o)$ for every $d \in D$, we have 
\begin{align*}
(\MAX \circ \AVG)(w|I) 
&= \frac{1}{n_{d^*}}\sum_{i \in N_{d^*}} \delta(i,w) \\
&\leq \frac{1}{n_{d^*}}\sum_{i \in N_{d^*}} \bigg( \delta(i,o) + \delta(w,o) \bigg) \\
&\leq (\MAX \circ \AVG)(o|I) + \delta(w,o). 
\end{align*}
Now, let $d_w$ be the district whose representative is $w$. In other words, $w$ is the rightmost $\lambda$-acceptable alternative for $d_w$. Let $o_w$ be the alternative that minimizes the total distance of the agents in $d_w$. We make the following two observations:
\begin{itemize}
\item Since $o_w$ is trivially $\lambda$-acceptable, it must be the case that she is located to the left of $w$, and thus $\delta(w,o) \leq \delta(o_w,o)$.
\item Since $o$ is to the right of $w$ and is not the representative of $d_w$, it must be the case that she is not $\lambda$-acceptable for $d_w$, that is, $\sum_{i \in N_{d_w}} \delta(i,o) > \lambda \sum_{i \in N_{d_w}} \delta(i,o_w)$. 
\end{itemize}
Combining these two observations together with the triangle inequality, we obtain
\begin{align*}
\delta(w,o) 
&\leq \delta(o,o_w)  
= \frac{1}{n_{d_w}} \sum_{i \in N_{d_w}} \delta(o,o_w)  \\
&\leq \frac{1}{n_{d_w}} \sum_{i \in N_{d_w}} \bigg( \delta(i,o) + \delta(i,o_w) \bigg) \\
&\leq \bigg( 1 + \frac{1}{\lambda} \bigg) \cdot \frac{1}{n_{d_w}} \sum_{i \in N_{d_w}} \delta(i,o) \\
&\leq \bigg( 1 + \frac{1}{\lambda} \bigg) \cdot (\MAX \circ \AVG)(o|I).
\end{align*}
Putting everything together, we obtain a distortion of at most $2 + \frac{1}{\lambda}$.

\medskip

\noindent 
{\bf Case 2: $w$ is to the right of $o$.} 
As in the previous case, let $d^*$ be the district with maximum cost for $w$. Denote by $y^*$ the representative of $d^*$, and by $o^*$ the optimal alternative for $d^*$ that minimizes the (average) total distance of the agents in $d^*$. We consider the following two subcases:
\begin{itemize}
\item 
{\bf $w$ is not $\lambda$-acceptable for $d^*$.} 
To argue about this case, we will need the following folklore technical lemma for instances on the line:

\begin{lemma} \label{lem:sum-monotonicity}
Let $S$ be a set of agents, and denote by $o_S$ the optimal alternative for the agents in $S$ that minimizes their total distance. Then, for every two alternatives $x$ and $y$, such that $y < x < o_S$, or $o_S < x < y$, it holds that $\sum_{i \in S} \delta(i,x) \leq \sum_{i \in S} \delta(i,y)$.
\end{lemma}

The fact that $w$ is the leftmost district representative implies that $y^*$ is to the right of $w$. 
In case $o^* < w < y$, Lemma~\ref{lem:sum-monotonicity} would yield that
\begin{align*}
\sum_{i \in N_{d^*}} \delta(i,w) \leq \sum_{i \in N_{d^*}} \delta(i,y^*) \leq \alpha \sum_{i \in N_{d^*}} \delta(i,o^*),
\end{align*}
thus contradicting that $w$ is not $\lambda$-acceptable for $d^*$. Hence, $o < w < o^*$, and Lemma~\ref{lem:sum-monotonicity} implies that
\begin{align*}
(\MAX \circ \AVG)(w|I) = \frac{1}{n_{d^*}}\sum_{i \in N_{d^*}} \delta(i,w) \leq \frac{1}{n_{d^*}} \sum_{i \in N_{d^*}} \delta(i,o) \leq (\MAX \circ \AVG)(o|I),
\end{align*}
meaning that $w$ is no worse than the optimal alternative in this case.

\item {\bf $w$ is $\lambda$-acceptable for $d^*$.} In this case, we clearly have that 
\begin{align*}
(\MAX \circ \AVG)(w|I)
&= \frac{1}{n_{d^*}} \sum_{i\in N_{d^*}} \delta(i,w) \\
&\leq \lambda \cdot \frac{1}{n_{d^*}} \sum_{i \in N_{d^*}} \delta(i,o^*) \\
&\leq \lambda \cdot \frac{1}{n_{d^*}} \sum_{i\in N_{d^*}} \delta(i,o) \\
&\leq \lambda \cdot (\MAX \circ \AVG)(o|I),
\end{align*} 
and thus the distortion is at most $\lambda$ in this case.
\end{itemize}
Putting everything together, the distortion of the mechanism is at most $\max\left\{2+ \frac{1}{\lambda}, \lambda \right\}$.
\end{proof}

By optimizing the distortion bound achieved by the class of $\lambda$-ARL mechanisms over the parameter $\lambda$, we see that the best such mechanism achieves a distortion of at most $1+\sqrt{2}$.

\begin{corollary} \label{cor:max-of-g-line}
For any $G \in \{\AVG,\MAX\}$, the $(\MAX \circ G)$-distortion of ($1+\sqrt{2}$)-ARL is at most $1+\sqrt{2}$.
\end{corollary}

We conclude this section by presenting a lower bound of $1+\sqrt{2}$ on the $(\MAX \circ G)$-distortion of distributed mechanisms, which holds even when the metric is a line, thus showing that $(1+\sqrt{2})$-ARL is the best possible on the line.  

\begin{theorem} \label{thm:max-of-g-cardinal-lower}
For any $G \in \{\AVG,\MAX\}$, the $(\MAX \circ G)$-distortion of any distributed mechanism is at least $1+\sqrt{2}$, even when the metric is a line.
\end{theorem}

\begin{proof}
Suppose towards a contradiction that there exists a distributed mechanism $M$ with $(\MAX \circ G)$-distortion strictly smaller than $1+\sqrt{2}$. We will consider instances with two alternatives $a$ and $b$ that are located at $0$ and $2$, respectively. In addition, our instances will consist of single-agent districts so that $\AVG=\MAX$ therein, and thus the distortion bounds hold for any $G \in \{\AVG,\MAX\}$. We can assume without loss of generality that $M$ outputs $a$ as the winner whenever it is given as input an instance with two districts of the same size, such that both alternatives are representative of some district. 

First, consider an instance $I_1$ with a single district consisting of one agent that is located at $2-\sqrt{2}$. Clearly, we have that $(\MAX \circ G)(a|I_1)=2-\sqrt{2}$ and $(\MAX \circ G)(b|I_1) = \sqrt{2}$. Therefore, $M$ must choose $a$ as the district representative, and thus the overall winner; otherwise its distortion would be at least $\frac{\sqrt{2}}{2-\sqrt{2}}=1 + \sqrt{2}$. 

Next, consider an instance $I_2$ with a single district consisting of one agent that is located at $2+\sqrt{2}$. Here, we have that $(\MAX \circ G)(a|I_2)=2+\sqrt{2}$ and $(\MAX \circ G)(b|I_2) = \sqrt{2}$. Therefore, $M$ must choose $b$ as the district representative; otherwise its distortion would be at least $\frac{2+\sqrt{2}}{\sqrt{2}}=1 + \sqrt{2}$.

Finally, consider an instance $I_3$ with the following two districts:
\begin{itemize}
\item The first district is similar to the one in $I_1$ and consists of an agent that is located at $2-\sqrt{2}$. By the above discussion, the representative of this district is $a$.

\item The second district is similar to the one in $I_2$ and consists of an agent that is located at $2+\sqrt{2}$. By the above discussion, the representative of this district is $b$.
\end{itemize}
Since both alternatives are representative of some district and the districts have the same unit size, $M$ outputs $a$ as the final winner. However, it holds that $(\MAX \circ G)(a|I_3)=2+\sqrt{2}$ (realized by the agent in the second district) and $(\MAX \circ G)(b|I_3)=\sqrt{2}$ (realized by both agents), and thus the distortion of $M$ is at least $\frac{2+\sqrt{2}}{\sqrt{2}} = 1+\sqrt{2}$, a contradiction. 
\end{proof}

\section{Extensions and Generalizations} \label{sec:extensions}

\subsection{Mechanisms that Select from the set of Representatives}
In the previous sections we looked at distributed mechanisms, which can choose any alternative as the final winner by essentially considering the district representatives as proxies. We now focus on the case where the final winner can only be chosen from among the district representatives (as in the work of \cite{FV2020facility}). To make the distinction between general mechanisms and those that select from the pool of district representatives clear, we will use the term {\em representative-selecting} to refer to the latter.

It is not hard to see that, with the exception of the bounds implied by Theorem~\ref{thm:main-composition} and its corollaries for general metric spaces (Corollaries~\ref{cor:main-cardinal} and~\ref{cor:main-ordinal}), the rest of our results follow by representative-selecting mechanisms. In particular:
\begin{itemize}
\item Every $\alpha$-in-arbitrary-over mechanism, as well as {\sc Arbitrary-Dictator}, choose some arbitrary representative (Theorem~\ref{thm:alpha-in-arbitrary-over}, Corollary~\ref{cor:5-max-of-G-ordinal}, and Theorem~\ref{thm:arbitrary-dictator});

\item The $1$-in-$1$-over cardinal mechanism and the $3$-in-$1$-over ordinal mechanism for $\AVG\circ G$ in the line metric, as well as {\sc Arbitrary-Median}, choose the median representative (Corollary~\ref{cor:main-cardinal} for line metric, Corollary~\ref{cor:SUM-G-line}, and Theorem~\ref{thm:arbitrary-median});

\item Every {\sc $\lambda$-ARL} mechanism chooses the leftmost representative (Theorem~\ref{thm:max-of-g-line-alpha} and Corollary~\ref{cor:max-of-g-line}). 
\end{itemize}
It is also not hard to see that all our lower bounds (Theorems~\ref{thm:avg-of-max-lower},~\ref{thm:max-of-max-lower},~\ref{thm:max-of-avg-lower}, and~\ref{thm:max-of-g-cardinal-lower}) also extend for the class of representative-selecting mechanisms: some representative is always chosen as the final winner in all instances used in the constructions. Based on all of the above discussion, we have the following corollary, which collects the best distortion bounds for the different objectives we consider.
 
\begin{theorem}\label{thm:representative-extension}
We can form ordinal representative-selecting mechanisms with distortion at most $5$ for $\MAX \circ G$ and general metric spaces. When the metric space is a line, the worst-case distortion of ordinal mechanisms is exactly $7$ for $\AVG \circ \AVG$, between $2+\sqrt{5}$ and $5$ for $\AVG \circ \MAX$, exactly $3$ for $\MAX \circ \MAX$, and at least $2+\sqrt{5}$ for $\MAX \circ \AVG$. When the metric is a line, the distortion of cardinal representative-selecting mechanisms is exactly $3$ for $\AVG \circ G$, and exactly $1+\sqrt{2}$ for $\MAX \circ G$.
\end{theorem}

Now, let us see how choosing only from the district representatives affects the bounds implied by Theorem~\ref{thm:main-composition} for general metric spaces. For clarity, we focus on objectives of the form $\AVG \circ G$; our discussion can easily be adapted for objectives of the form $\MAX \circ G$. Let $M$ be a representative-selecting mechanism that uses some in-district and over-districts direct voting rules. Given an instance $I=(N,M,D,\delta)$, let $R=R_M(I)$ be the set of district representatives chosen by $M$, and denote by $w=M(I) \in R$ the final winner. Clearly, Theorem~\ref{thm:main-composition} would hold without any modifications if, for any instance $I$, $w$ satisfies inequality \eqref{eq:SS-beta-over-districts} in the proof of the theorem:
\begin{align*}
 \forall j \in A: \sum_{i \in R} \delta(i, w) \leq \beta \cdot \sum_{i \in R} \delta(i, j).
\end{align*}
However, the distortion guarantees of direct voting rules used by $M$ in and over the districts are usually only with respect to the set of alternatives from which they are allowed to choose. So, if $M$ uses an over-districts rule that has $\AVG$-distortion at most $\gamma$, we have that, for any instance $I$, $w$ is such that
\begin{align}\label{eq:repsonly}
 \forall j \in R: \sum_{i \in R} \delta(i, w) \leq \gamma \cdot \sum_{i \in R} \delta(i, j).
\end{align}
Inequality~\eqref{eq:repsonly} cannot directly substitute inequality~\eqref{eq:SS-beta-over-districts} in the proof of Theorem~\ref{thm:main-composition} as it may be the case that the optimal alternative is not included in the set of representatives. So, we need to understand the relation between $\beta$ and $\gamma$, and then use it to obtain a distortion bound for $M$.

\begin{lemma} \label{lem:beta-gamma}
For any mechanism $M$, it holds that $\beta \leq 2\gamma$. 
\end{lemma}

\begin{proof}
By definition, $\beta$ is the worst-case ratio when comparing the total distance of the representatives from the outcome of $M$ to their total distance from the optimal {\em alternative}, whereas $\gamma$ is the worst-case ratio
when comparing the total distance of the representatives from the outcome of $M$ to their total distance from the optimal {\em representative}. So, to prove the statement, it suffices to show that 
\begin{align*}
\min_{j \in R} \sum_{i \in R} \delta(i,j) \leq 2 \cdot \min_{j \in A} \sum_{i \in R} \delta(i,j),
\end{align*}
for every instance that results in the set of district representatives $R$. This inequality holds trivially when there is only one district, so we assume that $k\geq 2$. To simplify notation, denote by $r \in R$ the optimal representative and by $o \in A$ the optimal alternative. For every $j \in R$, we have
$\sum_{i \in R} \delta(i,r) \leq \sum_{i \in R} \delta(i,j)$.
So, summing over all representatives, and using the triangle inequality, we obtain
\begin{align*}
\sum_{i \in R} \delta(i,r) 
&\leq \frac{1}{k} \sum_{j \in R} \sum_{i \in R} \delta(i,j) \\
&= \frac{2}{k} \sum_{i,j \in R} \delta(i,j) \\
&\leq \frac{2}{k} \sum_{i,j \in R} \bigg( \delta(i,o) + \delta(j,o) \bigg) \\
&= \frac{2(k-1)}{k} \sum_{i \in R} \delta(i,o).
\end{align*}
Thus, we have that $\beta \le \frac{2(k-1)}{k} \gamma \leq 2 \gamma$, and the proof is complete. 
\end{proof}

Due to Lemma~\ref{lem:beta-gamma}, Theorem~\ref{thm:main-composition} implies the following distortion bounds for general metric spaces and $\alpha$-in-$\gamma$-over representative-selecting mechanisms. 

\begin{theorem} \label{thm:representative-composition}
For general metric spaces and any $F,G\in \{\AVG,\MAX\}$, the $(F \circ G)$-distortion of any $\alpha$-in-$\gamma$-over representative-selecting mechanism is at most $\alpha+2\gamma+2\alpha\gamma$. 
\end{theorem}

\noindent 
Using appropriate direct voting rules in and over the districts, we can now again obtain concrete upper bounds on the distortion of cardinal and ordinal representative-selecting mechanisms. In particular, for every objective $F\circ G$, similarly to the case of general mechanisms in Section~\ref{sec:metric}, there is a $1$-in-$1$-over cardinal mechanism and a $3$-in-$2$-over ordinal mechanism. Therefore, by Theorem~\ref{thm:representative-composition}, we obtain the following corollary; note that Theorem~\ref{thm:representative-extension} implies better ordinal bounds for objectives of the form $\MAX \circ G$. This corollary shows that some loss in distortion may be experienced by forcing mechanisms to select a winner from only among the set of representatives, but that loss is not too large.

\begin{corollary}
For general metric spaces and any $F, G \in \{\AVG,\MAX\}$, there is a representative-selecting mechanism with $(F \circ G)$-distortion at most $5$, and an ordinal representative-selecting mechanism with $(F \circ G)$-distortion at most $19$.
\end{corollary}

\subsection{More General Objectives} \label{sec:generalization}

We now consider again mechanisms that can choose the final winner from the set of all alternatives, and discuss how some of our results (in particular, Theorems~\ref{thm:main-composition} and~\ref{thm:max-of-g-line-alpha}) can be extended for objectives $F\circ G$ beyond the cases where $F, G \in \{\AVG,\MAX\}$.

\subsubsection{Generalizing Theorem~\ref{thm:main-composition}} \label{sec:generalization1}

We previously showed in Theorem \ref{thm:main-composition} that the $(F\circ G)$-distortion of distributed mechanisms can be bounded in terms of the $F$- and $G$-distortion of the voting rules used in and over the districts. Here, we show that this theorem still holds for a much more general class of functions. To define this properly, we should think of $F$ and $G$ as functions that take as input vectors of distances. More precisely, given an instance $I=(N,A,D,\delta)$, let $f$ and $g$ be functions so that the {\em cost} of any alternative $j \in A$ is 
\begin{align*}
(F\circ G)(j|I) = f \bigg( g \Big( \vec{\delta}_{1}(j) \Big), \ldots, g \Big( \vec{\delta}_{k}(j) \Big) \! \bigg),    
\end{align*} 
where $\vec{\delta}_{d}(j)$ is the vector consisting of the distances $\delta(i,j)$ between every agent $i \in N_d$ and alternative $j$. To give a few examples, $g\Big(\vec{\delta}_{d}(j)\Big)=\frac{1}{n_d}\sum_{i\in N_d} \delta(i,j)$ if $G=\AVG$, and $g\Big(\vec{\delta}_{d}(j)\Big)=\max_{i\in N_d} \delta(i,j)$ if $G=\MAX$. More generally, we consider functions $f$ and $g$ which satisfy the following properties:
\begin{itemize}
\item {\bf Monotonicity}: 
A function $f$ is {\em monotone} if $f(\vec{v})\le f(\vec{u})$, for any two vectors $\vec{v}$ and $\vec{u}$  such that $v_\ell \leq u_\ell$ for every index $\ell$.

\item {\bf Subadditivity}:
A function $f$ is {\em subadditive} if $f(\vec{v}+\vec{u})\le f(\vec{v})+f(\vec{u})$, for any two vectors $\vec{v}$ and $\vec{u}$. Moreover, for any scalar $c \ge 1$, it must be that $f(c \cdot \vec{v})\leq c \cdot f(\vec{v})$.\footnote{The latter condition, sometimes known as sub-homogeneity, is not usually included in the standard definition of subadditive functions. It is easily implied by the first subadditivity condition when $c$ is an integer.}

\item {\bf Consistency}: 
A function $f$ is {\em consistent} if $f(\vec{v}) = c$, for any vector $\vec{v}$ such that $v_\ell = c$ for every index $\ell$.
\end{itemize}
Note that both $\AVG$ and $\MAX$, as well as many other functions, obey all of the above properties. 

\begin{theorem}\label{thm:general-composition} 
The distortion of any $\alpha$-in-$\beta$-over mechanism is at most $\alpha+\beta+\alpha\beta$, for any objective $F\circ G$ defined by functions $f$ and $g$ which are monotone, subadditive, and consistent.
\end{theorem}

\begin{proof}
Consider an arbitrary $\alpha$-in-$\beta$-over mechanism $M$ and an arbitrary instance $I=(N,A,D,\delta)$. 
Denote by $y_d$ the representative of each district $d\in D$, by $w$ the final winner, and by $o$ the optimal alternative. In addition, for each $d \in D$, let $\vec{t}_d(j)$ be a vector of size $n_d$ with every element being equal to the distance $\delta(y_d,j)$.
Then, the cost of $w$ is
\begin{align}
(F\circ G)(w|I) 
&=f \bigg( g \Big( \vec{\delta}_{1}(w) \Big), \ldots, g \Big( \vec{\delta}_{k}(w) \Big)\!\bigg) \nonumber \\
&\le f \bigg( g \Big( \vec{\delta}_{1}(y_1)+\vec{t}_1(w) \Big), \ldots, g\Big( \vec{\delta}_{k}(y_k)+\vec{t}_k(w) \Big) \! \bigg) \nonumber \\
&\le f \bigg( g \Big( \vec{\delta}_1(y_1) \Big), \ldots, g\Big( \vec{\delta}_k(y_k) \Big) \! \bigg) 
+ f \bigg( g \Big( \vec{t}_1(w) \Big), \ldots, g \Big( \vec{t}_k(w) \Big) \! \bigg) \nonumber \\
&= f \bigg( g \Big( \vec{\delta}_1(y_1) \Big), \ldots, g\Big( \vec{\delta}_k(y_k) \Big) \! \bigg) 
+ f \bigg( \delta(y_1,w), \ldots, \delta(y_k,w) \! \bigg), \label{eq:generalization-1}
\end{align}
where the first inequality follows by the monotonicity of $f$ and $g$ together with the triangle inequality, the second inequality is due to subadditivity, and the last equality is due to the consistency of $g$.

Since $M$ is $\alpha$-in-$\beta$-over, we have that the cost of $w$ if the representatives were pseudo-agents is at most a factor $\beta$ worse than the cost of the optimal alternative $o$. That is, 
\begin{align*}
f \bigg( \delta(y_1,w), \ldots, \delta(y_k,w) \! \bigg) 
&\le \beta \cdot f \bigg( \delta(y_1,o), \ldots, \delta(y_k,o) \! \bigg) \\
&= \beta \cdot f \bigg( g \Big( \vec{t}_1(o) \Big), \ldots, g \Big( \vec{t}_k(o) \Big) \! \bigg).
\end{align*}
Since $\delta(y_d,o) \le \delta(y_d,i)+\delta(i,o)$ for each agent $i \in N_d$, and the fact that the functions are monotone and subadditive, we obtain
\begin{align*}
f \bigg( g \Big( \vec{t}_1(o) \Big), \ldots, g \Big( \vec{t}_k(o) \Big) \! \bigg) 
&\le f \bigg( g \Big( \vec{\delta}_1(y_1) + \vec{\delta}_1(o) \Big), \ldots, g\Big( \vec{\delta}_k(y_k)+\vec{\delta}_k(o) \Big) \! \bigg) \\
&\le f \bigg( g\Big( \vec{\delta}_1(y_1) \Big), \ldots, g\Big(\vec{\delta}_k(y_k)\Big) \! \bigg) 
+ f \bigg( g \Big( \vec{\delta}_1(o) \Big), \ldots, g \Big( \vec{\delta}_k(o) \Big) \! \bigg).
\end{align*}
Using the above inequalities, \eqref{eq:generalization-1} yields
\begin{align}
(F\circ G)(w|I) 
&\le (1+\beta) \cdot f \bigg( g\Big( \vec{\delta}_1(y_1) \Big), \ldots, g\Big(\vec{\delta}_k(y_k)\Big) \! \bigg) + \beta \cdot f \bigg( g \Big( \vec{\delta}_1(o) \Big), \ldots, g \Big( \vec{\delta}_k(o) \Big) \! \bigg) \nonumber \\
&= (1+\beta) \cdot f \bigg( g\Big( \vec{\delta}_1(y_1) \Big), \ldots, g\Big(\vec{\delta}_k(y_k)\Big) \! \bigg) 
+ \beta \cdot (F\circ G)(o|I). \label{eq:generalization-2}
\end{align}

Since the in-district voting rule used by $M$ has $G$-distortion at most $\alpha$, we have that, for each district $d$, the cost of $y_d$ is at most $\alpha$ times the cost of $o$ for the agents in $d$, that is, 
$ g \Big( \vec{\delta}_d(y_d) \Big) \le \alpha\cdot g \Big( \vec{\delta}_d(o) \Big)$. Combining this with the second subadditivity property of $f$, we obtain
\begin{align*}
f \bigg( g \Big(\vec{\delta}_1(y_1)), \ldots, g \Big( \vec{\delta}_k(y_k) \Big) \! \bigg)
\le \alpha\cdot f \bigg( g \Big(\vec{\delta}_1(o)), \ldots, g \Big( \vec{\delta}_k(o) \Big) \! \bigg)
= \alpha \cdot (F\circ G)(o|I).
\end{align*}
By this, inequality \eqref{eq:generalization-2} now gives the desired upper bound of $\alpha + \beta + \alpha \beta$ on the distortion of $M$. 
\end{proof}

\subsubsection{Generalizing Theorem~\ref{thm:max-of-g-line-alpha}}
In this section we will show that, when the metric space is a line, Theorem \ref{thm:max-of-g-line-alpha} holds for more general objectives of the form $\MAX\circ G$. In particular, we are aiming to minimize the maximum cost of any district, which for an alternative $j \in A$ is given by some function $g\Big(\vec{\delta}_d(j)\Big)$ of the vector $\vec{\delta}_d(j)$ containing the distances between the members of $d$ and $j$. As in Section~\ref{sec:generalization1}, we will assume that $g$ is monotone, subadditive, and consistent. In addition, we require that $g$ is {\em single-peaked}, that is, for any district $d$, there is a there is a unique alternative $j$ that minimizes $g\Big(\vec{\delta}_d(j)\Big)$, and $g$ increases monotonically as we move further away from the location of $j$ (to the left or the right). It is easy to see that many functions, including $\AVG$ and $\MAX$, obey these properties.

As in Section~\ref{sec:line-cardinal}, the upper bound on the distortion is due to the $\lambda$-ARL mechanism, which chooses the representative of each district to be the rightmost $\lambda$-acceptable alternative for the district, and then outputs the leftmost representative as the final winner. Recall that the set of $\lambda$-acceptable alternatives for a district $d$ contains all the alternatives $x$ such that $g\Big(\vec{\delta}_d(x)\Big) \leq \lambda \cdot \min_{j \in A} g\Big(\vec{\delta}_d(j)\Big)$.

\begin{theorem} \label{thm:alpha-ARL-generalization}
The distortion of $\lambda$-ARL is at most $\max\left\{ 2 + \frac{1}{\lambda}, \lambda \right\}$ for any objective of the form $\MAX \circ G$, where $G$ is implemented by a monotone, subadditive, consistent, and single-peaked function $g$.
\end{theorem}

\begin{proof}
Consider an arbitrary instance $I = (N,A,D,\delta)$, where $\delta$ is a line metric. Let $w$ be the alternative chosen by $\lambda$-ARL when given $I$ as input, and denote by $o$ the optimal alternative. We consider the following two cases depending on the relative positions of $w$ and $o$.

\medskip

\noindent 
{\bf Case 1: $w$ is to the left of $o$.} 
Let $d^* \in \arg\max_{d \in D} g\Big(\vec{\delta}_d(w)\Big)$ be the district that gives the max cost for $w$.  Also, for every $d \in D$, let $\vec{\delta}_d(w,o)$ be a vector of size $n_d$ with all its entries equal to $\delta(w,o)$.
By the triangle inequality, the monotonicity, subadditivity and consistency of $g$, as well as the fact that $(\MAX\circ G)(o|I) \geq g\Big(\vec{\delta}_d(o)\Big)$ for any district $d$, we have
\begin{align*}
(\MAX \circ G)(w|I) 
= g\Big(\vec{\delta}_{d^*}(w)\Big) 
&\leq g\Big(\vec{\delta}_{d^*}(o)+\vec{\delta}_{d^*}(w,o)\Big) \\
&\leq g\Big(\vec{\delta}_{d^*}(o)\Big)  + g\Big(\vec{\delta}_{d^*}(w,o)\Big)\\
&\leq  (\MAX\circ G)(o|I) + \delta(w,o).
\end{align*}
Now, let $d_w$ be the district whose representative is $w$. In other words, $w$ is the rightmost $\lambda$-acceptable alternative for $d_w$. Let $o_w$ be the alternative with minimum cost according to $g$ for the agents in $d_w$. We make the following two observations:
\begin{itemize}
\item 
Since $o_w$ is trivially $\lambda$-acceptable, it must be the case that she is located to the left of $w$, and thus $\delta(w,o) \leq \delta(o_w,o)$.

\item 
Since $o$ is to the right of $w$ and is not the representative of $d_w$, it must be the case that she is not $\lambda$-acceptable for $d_w$, that is, $g\Big(\vec\delta_{d_w}(o)\Big) > \lambda \cdot g\Big(\vec\delta_{d_w}(o_w)\Big)$. 
\end{itemize}
Combining these two observations together with the triangle inequality and the properties of $g$ (monotonicity, subadditivity, and consistency), we obtain
\begin{align*}
\delta(w,o)
&\leq \delta(o,o_w)
= g\Big(\vec{\delta}_{d_w}(o,o_w)\Big) \\
&\leq g\Big(\vec{\delta}_{d_w}(o) + \vec{\delta}_{d_w}(o_w)\Big) \\
&\leq g\Big(\vec{\delta}_{d_w}(o)\Big) + g\Big(\vec{\delta}_{d_w}(o_w)\Big) \\
&\leq \bigg( 1 + \frac{1}{\lambda} \bigg) \cdot g(\vec{\delta}_{d_w}(o)) \\
&\leq \bigg( 1 + \frac{1}{\lambda} \bigg) \cdot (\MAX \circ G)(o|I).
\end{align*}
Putting everything together, we obtain a distortion of at most $2 + \frac{1}{\lambda}$.

\medskip

\noindent 
{\bf Case 2: $w$ is to the right of $o$.} 
As in the previous case, let $d^*$ be the district with maximum cost for $w$. Denote by $y^*$ the representative of $d^*$, and by $o^*$ the optimal alternative for $d^*$ that minimizes the cost of the agents in $d^*$ according to $g$. We consider the following two subcases:
\begin{itemize}
\item 
{\bf $w$ is not $\lambda$-acceptable for $d^*$.} 
The fact that $w$ is the leftmost district representative implies that $y^*$ is to the right of $w$. 
If $w$ is between $o^*$ and $y^*$, then the fact that $g$ is single-peaked would yield that
\begin{align*}
g\Big(\vec{\delta}_{d^*}(w)\Big) \leq g\Big(\vec{\delta}_{d^*}(y^*)\Big) \leq \lambda \cdot g\Big(\vec{\delta}_{d^*}(o^*)\Big),
\end{align*}
thus contradicting that $w$ is not $\lambda$-acceptable for $d^*$ in this subcase. Hence, it must be the case that $w$ is between $o$ and $o^*$, and since $g$ is single-peaked, we have that
\begin{align*}
(\MAX \circ G)(w|I) = g\Big(\vec{\delta}_{d^*}(w)\Big) \leq g\Big(\vec{\delta}_{d^*}(o)\Big) \leq (\MAX \circ G)(o|I).
\end{align*}

\item {\bf $w$ is $\lambda$-acceptable for $d^*$.} In this case, we clearly have that 
\begin{align*}
(\MAX \circ G)(w|I)
&= g\Big(\vec{\delta}_{d^*}(w)\Big) \\
&\leq \lambda \cdot g\Big(\vec{\delta}_{d^*}(o^*)\Big) \\
&\leq \lambda \cdot g\Big(\vec{\delta}_{d^*}(o)\Big) \\
&\leq \lambda \cdot (\MAX \circ G)(o|I).
\end{align*} 
\end{itemize}
Putting everything together, we obtain an upper bound of $\max\left\{2 + \frac{1}{\lambda}, \lambda\right\}$.
\end{proof}

By optimizing over the parameter $\lambda$, we obtain the following generalization of Corollary~\ref{cor:max-of-g-line}.

\begin{corollary} \label{cor:max-of-g-line-generalization}
The distortion of ($1+\sqrt{2}$)-ARL is at most $1+\sqrt{2}$ for any objective of the form $\MAX \circ G$, where $G$ is defined by a monotone, subadditive, consistent, and single-peaked function $g$.
\end{corollary}

\section{Open Problems} \label{sec:open}
In this paper, we showed bounds on the distortion of single-winner distributed mechanisms for many different objectives, some of which are novel and make sense only in this particular setting. Still, there are several challenging open questions, as well as new directions for future research. Starting with our results, it would be interesting to close the gaps between the lower and upper bounds presented in Table~\ref{table:results} for the various scenarios we considered. For cases where our bounds for general metrics and the line differ significantly, such as for ordinal mechanisms and the $\AVG \circ \MAX$ objective, one could focus on other well-structured metrics, like the Euclidean space or generalizations of it. 

Since we focused exclusively on deterministic mechanisms, a possible direction could be to consider randomized mechanisms and investigate whether better distortion bounds are possible. Note that our composition theorem (Theorem~\ref{thm:main-composition} and its variants) already provide randomized bounds by plugging in appropriate randomized in-district and over-districts direct voting rules. However, these bounds seem extremely loose, and different techniques are required to obtain tight bounds. Going beyond the single-winner setting, one could study the distortion of distributed mechanisms that output committees of a given number of alternatives, or rankings of all alternatives. Finally, another interesting direction would be to study what happens when agents act strategically, and either understand how this behavior affects given distributed mechanisms, or aim to design strategyproof mechanisms that are resilient to manipulation and at the same time achieve low distortion.

\bibliographystyle{plainnat}
\bibliography{references}

\begin{thebibliography}{31}
\providecommand{\natexlab}[1]{#1}
\providecommand{\url}[1]{\texttt{#1}}
\expandafter\ifx\csname urlstyle\endcsname\relax
  \providecommand{\doi}[1]{doi: #1}\else
  \providecommand{\doi}{doi: \begingroup \urlstyle{rm}\Url}\fi

\bibitem[Abramowitz and Anshelevich(2018)]{abramowitz2017utilitarians}
Ben Abramowitz and Elliot Anshelevich.
\newblock Utilitarians without utilities: Maximizing social welfare for graph
  problems using only ordinal preferences.
\newblock In \emph{Proceedings of the 32nd {AAAI} Conference on Artificial
  Intelligence ({AAAI})}, pages 894--901, 2018.

\bibitem[Amanatidis et~al.(2021{\natexlab{a}})Amanatidis, Birmpas,
  Filos{-}Ratsikas, and Voudouris]{amanatidis2020peeking}
Georgios Amanatidis, Georgios Birmpas, Aris Filos{-}Ratsikas, and Alexandros~A.
  Voudouris.
\newblock Peeking behind the ordinal curtain: Improving distortion via cardinal
  queries.
\newblock \emph{Artificial Intelligence}, 296:\penalty0 103488,
  2021{\natexlab{a}}.

\bibitem[Amanatidis et~al.(2021{\natexlab{b}})Amanatidis, Birmpas,
  Filos-Ratsikas, and Voudouris]{amanatidis2021matching}
Georgios Amanatidis, Georgios Birmpas, Aris Filos-Ratsikas, and Alexandros~A.
  Voudouris.
\newblock A few queries go a long way: Information-distortion tradeoffs in
  matching.
\newblock In \emph{Proceedings of the 35th {AAAI} Conference on Artificial
  Intelligence ({AAAI})}, pages 5078--5085, 2021{\natexlab{b}}.

\bibitem[Anagnostides et~al.(2021)Anagnostides, Fotakis, and
  Patsilinakos]{anagnostides2021metric}
Ioannis Anagnostides, Dimitris Fotakis, and Panagiotis Patsilinakos.
\newblock Metric-distortion bounds under limited information.
\newblock \emph{CoRR}, abs/2107.02489, 2021.

\bibitem[Anshelevich and Postl(2017)]{anshelevich2017randomized}
Elliot Anshelevich and John Postl.
\newblock Randomized social choice functions under metric preferences.
\newblock \emph{Journal of Artificial Intelligence Research}, 58:\penalty0
  797--827, 2017.

\bibitem[Anshelevich and Zhu(2018)]{anshelevich2018ordinal}
Elliot Anshelevich and Wennan Zhu.
\newblock Ordinal approximation for social choice, matching, and facility
  location problems given candidate positions.
\newblock In \emph{Proceedings of the 14th International Conference on Web and
  Internet Economics ({WINE})}, pages 3--20, 2018.

\bibitem[Anshelevich et~al.(2018)Anshelevich, Bhardwaj, Elkind, Postl, and
  Skowron]{anshelevich2018approximating}
Elliot Anshelevich, Onkar Bhardwaj, Edith Elkind, John Postl, and Piotr
  Skowron.
\newblock Approximating optimal social choice under metric preferences.
\newblock \emph{Artificial Intelligence}, 264:\penalty0 27--51, 2018.

\bibitem[Anshelevich et~al.(2021)Anshelevich, Filos{-}Ratsikas, Shah, and
  Voudouris]{survey}
Elliot Anshelevich, Aris Filos{-}Ratsikas, Nisarg Shah, and Alexandros~A.
  Voudouris.
\newblock Distortion in social choice problems: The first 15 years and beyond.
\newblock In \emph{Proceedings of the 30th International Joint Conference on
  Artificial Intelligence ({IJCAI})}, 2021.

\bibitem[Benad{\`{e}} et~al.(2017)Benad{\`{e}}, Nath, Procaccia, and
  Shah]{benade2017participatory}
Gerdus Benad{\`{e}}, Swaprava Nath, Ariel~D. Procaccia, and Nisarg Shah.
\newblock Preference elicitation for participatory budgeting.
\newblock In \emph{Proceedings of the 31st {AAAI} Conference on Artificial
  Intelligence ({AAAI})}, pages 376--382, 2017.

\bibitem[Benad{\`{e}} et~al.(2019)Benad{\`{e}}, Procaccia, and
  Qiao]{benade2019rankings}
Gerdus Benad{\`{e}}, Ariel~D. Procaccia, and Mingda Qiao.
\newblock Low-distortion social welfare functions.
\newblock In \emph{Proceedings of the 33rd {AAAI} Conference on Artificial
  Intelligence ({AAAI})}, pages 1788--1795, 2019.

\bibitem[Bhaskar et~al.(2018)Bhaskar, Dani, and Ghosh]{bhaskar2018truthful}
Umang Bhaskar, Varsha Dani, and Abheek Ghosh.
\newblock Truthful and near-optimal mechanisms for welfare maximization in
  multi-winner elections.
\newblock In \emph{Proceedings of the 32nd {AAAI} Conference on Artificial
  Intelligence ({AAAI})}, pages 925--932, 2018.

\bibitem[Borodin et~al.(2019)Borodin, Lev, Shah, and
  Strangway]{borodin2019primaries}
Allan Borodin, Omer Lev, Nisarg Shah, and Tyrone Strangway.
\newblock Primarily about primaries.
\newblock In \emph{Proceedings of the 33rd {AAAI} Conference on Artificial
  Intelligence ({AAAI})}, pages 1804--1811, 2019.

\bibitem[Boutilier et~al.(2015)Boutilier, Caragiannis, Haber, Lu, Procaccia,
  and Sheffet]{boutilier2015optimal}
Craig Boutilier, Ioannis Caragiannis, Simi Haber, Tyler Lu, Ariel~D. Procaccia,
  and Or~Sheffet.
\newblock Optimal social choice functions: A utilitarian view.
\newblock \emph{Artificial Intelligence}, 227:\penalty0 190--213, 2015.

\bibitem[Caragiannis and Procaccia(2011)]{caragiannis2011embedding}
Ioannis Caragiannis and Ariel~D. Procaccia.
\newblock Voting almost maximizes social welfare despite limited communication.
\newblock \emph{Artificial Intelligence}, 175\penalty0 (9-10):\penalty0
  1655--1671, 2011.

\bibitem[Caragiannis et~al.(2017)Caragiannis, Nath, Procaccia, and
  Shah]{caragiannis2017subset}
Ioannis Caragiannis, Swaprava Nath, Ariel~D. Procaccia, and Nisarg Shah.
\newblock Subset selection via implicit utilitarian voting.
\newblock \emph{Journal of Artificial Intelligence Research}, 58:\penalty0
  123--152, 2017.

\bibitem[Chen et~al.(2020)Chen, Li, and Wang]{chen2020favorite}
Xujin Chen, Minming Li, and Chenhao Wang.
\newblock Favorite-candidate voting for eliminating the least popular candidate
  in a metric space.
\newblock In \emph{Proceedings of the 34th {AAAI} Conference on Artificial
  Intelligence ({AAAI})}, pages 1894--1901, 2020.

\bibitem[Fain et~al.(2019)Fain, Goel, Munagala, and Prabhu]{fain2019random}
Brandon Fain, Ashish Goel, Kamesh Munagala, and Nina Prabhu.
\newblock Random dictators with a random referee: Constant sample complexity
  mechanisms for social choice.
\newblock In \emph{Proceedings of the 33rd {AAAI} Conference on Artificial
  Intelligence ({AAAI})}, pages 1893--1900, 2019.

\bibitem[Feldman et~al.(2016)Feldman, Fiat, and Golomb]{feldman2016voting}
Michal Feldman, Amos Fiat, and Iddan Golomb.
\newblock On voting and facility location.
\newblock In \emph{Proceedings of the 2016 {ACM} Conference on Economics and
  Computation ({EC})}, pages 269--286, 2016.

\bibitem[Filos{-}Ratsikas and Voudouris(2021)]{FV2020facility}
Aris Filos{-}Ratsikas and Alexandros~A. Voudouris.
\newblock Approximate mechanism design for distributed facility location.
\newblock In \emph{Proceedings of the 14th International Symposium on
  Algorithmic Game Theory ({SAGT})}, 2021.

\bibitem[Filos{-}Ratsikas et~al.(2014)Filos{-}Ratsikas, Frederiksen, and
  Zhang]{filos-Ratsikas2014matching}
Aris Filos{-}Ratsikas, S{\o}ren Kristoffer~Stiil Frederiksen, and Jie Zhang.
\newblock Social welfare in one-sided matchings: Random priority and beyond.
\newblock In \emph{Proceedings of the 7th Symposium on Algorithmic Game Theory
  ({SAGT})}, pages 1--12, 2014.

\bibitem[Filos-Ratsikas et~al.(2020)Filos-Ratsikas, Micha, and
  Voudouris]{FMV2020distributed}
Aris Filos-Ratsikas, Evi Micha, and Alexandros~A. Voudouris.
\newblock The distortion of distributed voting.
\newblock \emph{Artificial Intelligence}, 286:\penalty0 103343, 2020.

\bibitem[Gkatzelis et~al.(2020)Gkatzelis, Halpern, and
  Shah]{gkatzelis2020resolving}
Vasilis Gkatzelis, Daniel Halpern, and Nisarg Shah.
\newblock Resolving the optimal metric distortion conjecture.
\newblock In \emph{Proceedings of the 61st {IEEE} Annual Symposium on
  Foundations of Computer Science ({FOCS})}, pages 1427--1438, 2020.

\bibitem[Goel et~al.(2017)Goel, Krishnaswamy, and Munagala]{goel2017metric}
Ashish Goel, Anilesh~K Krishnaswamy, and Kamesh Munagala.
\newblock Metric distortion of social choice rules: Lower bounds and fairness
  properties.
\newblock In \emph{Proceedings of the 2017 {ACM} Conference on Economics and
  Computation (EC)}, pages 287--304, 2017.

\bibitem[Jaworski and Skowron(2020)]{jaworski20evaluating}
Michal Jaworski and Piotr Skowron.
\newblock Evaluating committees for representative democracies: The distortion
  and beyond.
\newblock In \emph{Proceedings of the 29th International Joint Conference on
  Artificial Intelligence ({IJCAI})}, pages 196--202, 2020.

\bibitem[Kempe(2020{\natexlab{a}})]{kempe2020communication}
David Kempe.
\newblock Communication, distortion, and randomness in metric voting.
\newblock In \emph{Proceedings of the 34th {AAAI} Conference on Artificial
  Intelligence ({AAAI})}, pages 2087--2094, 2020{\natexlab{a}}.

\bibitem[Kempe(2020{\natexlab{b}})]{kempe2020duality}
David Kempe.
\newblock An analysis framework for metric voting based on {LP} duality.
\newblock In \emph{Proceedings of the 34th {AAAI} Conference on Artificial
  Intelligence ({AAAI})}, pages 2079--2086, 2020{\natexlab{b}}.

\bibitem[Mandal et~al.(2019)Mandal, Procaccia, Shah, and
  Woodruff]{mandal2019thrifty}
Debmalya Mandal, Ariel~D. Procaccia, Nisarg Shah, and David~P. Woodruff.
\newblock Efficient and thrifty voting by any means necessary.
\newblock In \emph{Proceedings of the 32nd Annual Conference on Neural
  Information Processing Systems ({NeurIPS})}, pages 7178--7189, 2019.

\bibitem[Mandal et~al.(2020)Mandal, Shah, and Woodruff]{mandal2020optimal}
Debmalya Mandal, Nisarg Shah, and David~P. Woodruff.
\newblock Optimal communication-distortion tradeoff in voting.
\newblock In \emph{Proceedings of the 21st {ACM} Conference on Economics and
  Computation ({EC})}, pages 795--813, 2020.

\bibitem[Munagala and Wang(2019)]{munagala2019improved}
Kamesh Munagala and Kangning Wang.
\newblock Improved metric distortion for deterministic social choice rules.
\newblock In \emph{Proceedings of the 2019 {ACM} Conference on Economics and
  Computation ({EC})}, pages 245--262, 2019.

\bibitem[Procaccia and Rosenschein(2006)]{procaccia2006distortion}
Ariel~D. Procaccia and Jeffrey~S. Rosenschein.
\newblock The distortion of cardinal preferences in voting.
\newblock In \emph{International Workshop on Cooperative Information Agents
  ({CIA})}, pages 317--331, 2006.

\bibitem[Sen(1986)]{sen1986social}
Amartya Sen.
\newblock Social choice theory.
\newblock \emph{Handbook of mathematical economics}, 3:\penalty0 1073--1181,
  1986.

\end{thebibliography}

\end{document}